\documentclass{article}

\usepackage[english]{babel}

\usepackage[letterpaper,top=2cm,bottom=2cm,left=3cm,right=3cm,marginparwidth=1.75cm]{geometry}

\usepackage{amsmath}
\usepackage{graphicx}
\usepackage{float}
\usepackage{tikz}
\usepackage{amssymb}
\usepackage{caption} 
\usepackage{subcaption} 
\usepackage[colorlinks=true, allcolors=blue]{hyperref}
\usepackage{amsthm} 
\usepackage{mathtools} 
\usepackage{enumitem}
\usepackage{indentfirst}
\newcommand{\R}{\mathbb R}
\newcommand{\mc}{\mathcal}

\newtheorem{theorem}{Theorem}[section]  
\newtheorem{definition}{Definition}[section] 
\newtheorem{lemma}{Lemma}[section] 
\newtheorem{proposition}[theorem]{Proposition}
\newtheorem{remark}{Remark}[section]

\title{A Minimal Dynamical Model for Incubation–Outbreak Transitions in Social Norm Diffusion}
\author{Run Wang, Leonardo Cianfanelli, Giacomo Como, Wenjun Mei}

\begin{document}
\maketitle

\begin{abstract}
In this paper, we introduce a minimal dynamical model for the diffusion of a new social norm, in which individuals transition among three states: non-supporters, silent supporters, and vocal advocates. Despite its simplicity and its close relation to the classical SI-type and SIS-type spreading dynamics, this model exhibits a nontrivial latent–outbreak dynamic pattern: an initial small adoption wave is followed by a long quiescent period and then an abrupt, endogenous explosion of support. Through a complete analytical characterization of equilibria, stability, convergence, and phase-transition conditions, we identify the nonlinear mechanism that generates this incubation phenomenon and derive estimates of the latent period. Our results reveal a dynamical route to sudden social change driven purely by internal interactions rather than external shocks.
\end{abstract}

\section{Introduction}
\subsection{Background}
Social changes refer to the alteration of mechanisms within the social structure, characterized by shifts in cultural symbols, behavioral norms, social organizations, or value systems. While such changes are often driven by collective behavior, their origins can be traced to individual-level mechanisms \cite{nowak1990private}. Social change on an individual level refers to a shift in a person's beliefs, values, or perspectives, and when a sufficient number of people embrace a new viewpoint, it leads to a social change.

Our work focuses on the diffusion of a new social norm among human groups. Social change often follows a distinct pattern driven by external factors (e.g., the occurrence of a social event) and social dynamics. Initially, the triggering event leads to a shift in perspectives among a subset of individuals. Within this group, a fraction actively engages in public discourse, advocating for the new viewpoint. Through sustained interaction, those initially resistant to the change gradually become persuaded by these vocal proponents. Concurrently, individuals who privately support the new perspective but previously remained silent begin participating in the discourse as its visibility and intensity increase. However, the willingness of vocal advocates to sustain their active participation tends to diminish over time.

\subsection{Literature Review}
\indent
The application of epidemiological frameworks to simulate the diffusion of ideas, innovations, and information is well-established in the social sciences. Research by Bettencourt et al. \cite{BETTENCOURT2006513} indicates that the population dynamics underlying the diffusion of ideas are qualitatively similar to those involved in the spread of infections, allowing for precise quantification using parameters such as contact rates and incubation periods. Traditional population-based growth models, most notably the Bass model \cite{916272ae-6b7a-3e62-b7e2-5747187dae7b}, provide a foundation for understanding the non-linear timing of initial adoption through the interplay of innovative and imitative behaviors. Beyond these macro-scale models, individual-based and spatially explicit frameworks by Ling Bian etc. \cite{doi:10.1068/b2833} emphasize how local and long-distance interactions among micro-level agents shape the rhythm of propagation. Such models have demonstrated significant cross-disciplinary utility, effectively capturing the transition outcomes in fields as diverse as computer virus suppression \cite{PIQUEIRA2009355}, social media influencer maximization \cite{MORE2019102}, and pathogen invasion strategies \cite{annurev:/content/journals/10.1146/annurev.phyto.45.062806.094357, doi:10.1177/0885728810387922}.

The phenomenon of an outbreak occurring after a prolonged period of dormancy is a signature of non-linear system dynamics. At the structural level, small-world effects can induce a critical transition from local oscillations to global synchronized outbreaks \cite{PhysRevLett.86.2909}, illustrating how network topology dictates spreading patterns. Mathematically, these abrupt shifts can be triggered by a Hopf bifurcation, where a system undergoes a qualitative transition from steady states to periodic outbreaks as specific parameters cross a threshold \cite{XIAO2001733}. Crucially, the presence of bistability in reinfection models allows for the existence of resurgent epidemics, providing a mechanistic explanation for how opinions can oscillate between dormant and active states \cite{8353159}. These dynamics are further complicated by behavioral feedback loops; game-theoretic models coupled with epidemiological frameworks suggest that individual risk perception and the cost of adherence lead to transitions between multiple, non-concurrently stable equilibria \cite{doi:10.1073/pnas.2311584120}. Furthermore, collective patterns of social diffusion are often shaped by the micro-level balance between individual inertia \cite{alos2016inertia} and trend-seeking behaviors \cite{ye2021collective}, which explains the lag effects often observed before a rapid burst of adoption. Sensitivity analysis remains a critical tool for identifying which of these parameters most significantly drive the system from a latent state toward a global cascade \cite{https://doi.org/10.1155/2013/721406}.

The mechanisms of social norm change are fundamentally driven by the non-linear relationship between micro-motives and macro-outcomes. Threshold models of collective behavior \cite{doi:10.1086/226707} specify that an individual's decision to adopt a new norm depends on the proportion of their peers who have already done so. This process is often hidden by "preference falsification" \cite{kuran1991now}, where individuals conceal their true intentions due to social pressure, causing immense latent pressure to accumulate until a minor trigger causes an explosive and unexpected revolution. Empirical evidence confirms the existence of tipping points, demonstrating that when a committed minority reaches a critical mass—typically around 25\%—it can spontaneously overturn an established consensus and establish a new social coordination equilibrium \cite{doi:10.1126/science.aas8827, doi:10.1073/pnas.1418838112}. This transition from private attitude to public opinion involves complex spatial self-organization \cite{nowak1990private} and kinetic exchange processes between internal beliefs and external expressions \cite{10.1098/rsta.2021.0169}. Historically, such rapid collapses of deeply rooted norms, like the abolition of footbinding \cite{brown2018economic}, highlight how economic and social correlates drive these swift transitions. Finally, the stagnation or inhibition of these processes can often be attributed to social network fatigue \cite{ravindran2014antecedents} or negative feedback from successful reforms \cite{jain_public_nodate}, illustrating the cyclical and fragile nature of social change.

\subsection{Statement of Contribution}

In this work, we provide a comprehensive analysis of a novel dynamic model that explicitly accounts for the interplay between activist fatigue and latent consensus accumulation. Specifically, our contributions are three-fold. First, we identify and characterize an 'incubation phenomenon'—a counter-intuitive dynamic where a decline in public discussion and active supporters is followed by a significant resurgence of attention and widespread acceptance of the new social norm. This provides a mechanistic explanation for how social systems 'store energy' during periods of apparent stagnation. Second, we provide a rigorous mathematical analysis of the proposed dynamical system. This includes the identification of equilibrium points, stability analysis under varying parameter regimes, and the examination of the system's asymptotic properties. Third, we establish the threshold conditions that govern the transition between different social states. We demonstrate how specific parameter configurations—such as the fatigue rate and the silence threshold—determine whether a system will undergo social change or fall into permanent stagnation. These results provide an alternative perspective on the unpredictable nature of social tipping points \cite{ravindran2014antecedents, kuran1991now}.

\subsection{Organization}
The rest of the paper is organized as follows. In Section \ref{sec:model}, we introduce the model setup, while Section \ref{sec:basic} presents the basic results regarding equilibrium points and their stability. In Section \ref{sec:H}, we provide an invariant of motion of the system and characterize the geometry of its level sets. Section~\ref{sec:beh} investigates the asymptotic and transient behavior of the system, providing a rigorous proof for the divergence of the incubation period and analyzing the conditions that dictate the monotonicity of the advocacy level. Finally, Section \ref{sec:conclusion} is the conclusion.

\section{Model Set-up and Notations}\label{sec:model}
We consider a large, well-mixed population where each individual is characterized by two distinct attributes: their internal belief and their public behavior regarding a newly introduced social norm. Let $x(t) \in [0, 1]$ denote the fraction of the population that has internally accepted the new norm at time $t$. Among these supporters, only a subset actively participates in public discourse. We denote by $y(t) \in [0, 1]$ the fraction of the population that publicly advocates for the norm. By definition, these variables must satisfy the physical constraint $y(t) \le x(t)$ for all $t \ge 0$, a property that will be proved to hold in the subsequent section to ensure the model's well-posedness.

The evolution of the social transition is governed by two coupled mechanisms: \emph{latent belief adoption} and \emph{behavioral expression and exhaustion}.

\paragraph{1) Latent Belief Adoption (SI-like Dynamics):} 
The expansion of the supporter base occurs as unconverted individuals encounter public advocates and are persuaded to accept the new norm. Drawing an analogy from classical epidemiological models, this process follows SI-like dynamics:
\begin{equation}
    \dot{x}(t) = \alpha \big(1 - x(t)\big) y(t)\,,
    \label{equ_x}
\end{equation}
where $\alpha > 0$ represents the \textit{conversion rate} or the persuasiveness of the advocates. This formulation implies that the growth rate of supporters is jointly determined by the pool of potential converts $(1-x)$ and the current intensity of public advocacy $(y)$. Notably, this transition is modeled as an irreversible process, reflecting the nature of deep-seated cognitive shifts or the internalization of new values.

\paragraph{2) Behavioral Expression and Exhaustion (SIS-like Dynamics):} 
While all $x$ individuals internally support the norm, their transition from silence to active advocacy, and their subsequent withdrawal due to exhaustion, is captured by an SIS-like dynamic:
\begin{equation}
    \dot{y}(t) = \beta \big(x(t) - y(t)\big) y(t) - \gamma y(t)\,,
    \label{equ_y}
\end{equation}
where $\beta > 0$ denotes the \textit{mobilization rate} and $\gamma > 0$ represents the \textit{decay rate}, characterizing the spontaneous loss of desire for public expression. Here, the term $(x-y)$ explicitly represents the pool of \emph{silent supporters}. Equation \eqref{equ_y} reflects the interplay between two opposing forces: the recruitment of silent supporters by active ones and the natural decay of advocacy due to individual fatigue.

\paragraph{3) Parameter Scaling and Simplification:} 
To facilitate the subsequent mathematical analysis, we observe that the system's qualitative behavior is invariant under a proportional scaling of all parameters. By rescaling the time variable $\tilde{t} = \alpha t$, we can set $\alpha = 1$ without loss of generality. This reduces the system to a minimal form
\begin{equation}
\label{sys_dynamics}
\begin{cases}
    \dot{x} = (1 - x)y\,, \\
    \dot{y} = \beta(x - y)y - \gamma y\,,
\end{cases}
\end{equation}
where the dynamics are now characterized by two fundamental dimensionless ratios: the mobilization-to-conversion ratio ($\beta/\alpha$) and the fatigue-to-conversion ratio ($\gamma/\alpha$). The diagram illustrating the transitions is reported in Figure %\ref{fig:block1}-
\ref{fig:block2}.

\begin{figure}[H]
    \centering
    \includegraphics[width=0.7\textwidth]{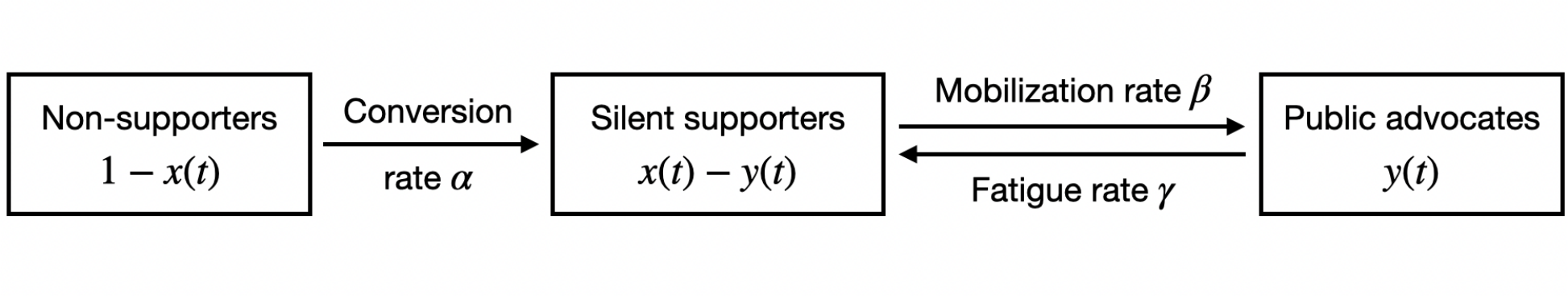}  
    \caption{A diagram illustrating the private-attitude conversion, the mobilization, and the fatigue mechanisms. \label{fig:block2}}
\end{figure}

\section{Basic Properties: Equilibrium Points and Their Stability}\label{sec:basic}

Before proceeding to the analysis of the system asymptotic behavior, it is essential to establish the mathematical well-posedness of the proposed model. Specifically, we must ensure that the state variables $(x, y)$, representing population fractions, remain within a physically meaningful domain for all $t \ge 0$. In the following proposition, we rigorously prove the positive invariance of the state space, demonstrating that any trajectory starting within the unit square remains confined therein, and further, that the structural constraint $y(t) \le x(t)$ is preserved under the system dynamics. Based on the meaning of x and y, we consider the following set $$A = \{(x, y)\in \mathbb{R}^2 : 0 \leq y \leq x \leq 1\}.$$

\begin{proposition}\label{proposition:basic}
Consider model \eqref{sys_dynamics}. Then: 
\begin{enumerate}
    \item[(i)] for every initial condition $(x(0), y(0))$ in $A$, there exists a unique solution $(x(t), y(t))$. The solution is $\mc C^\infty$ and belongs to $A$ for every $t$ in $\R_+$;
    \item[(ii)] $x(t)$ is monotonically non-decreasing.
\end{enumerate}
\end{proposition}

\begin{proof}
Since the right-hand side of \eqref{sys_dynamics} is $\mc C^{\infty}$, existence, local uniqueness, and regularity of the solution are standard \cite[Chapter~I.3]{hale}. To prove that $A$ is positively invariant, note that $y = 0$ implies $\dot y = 0$, so that $y(t) \ge 0$ for every time $t \ge 0$ if $y(0) \ge 0$. Since $x(t)$ is non-decreasing in time, $x(t) \ge 0$ if $x(0) \ge 0$. Moreover, $\dot x = 0$ if $x = 1$, hence $x(t)$ cannot exceed 1. To prove that $y(t) \le x(t)$ for every $t$, define $z(t) = x(t) - y(t) $ and note that $\dot z = (1 - x - \beta(x - y) + \gamma)y$ is non-negative whenever $x = y \le 1$, hence $y$ can never exceed $x$. This proves that $A$ is positively invariant. Since $A$ is compact, uniqueness of the solution is ensured globally, proving item (i). Item (ii) follows from the fact that $A$ is positively invariant, hence $\dot x(t) = x(t) (1-y(t)) \ge 0$ for every $t \ge 0$.
\end{proof}

Having established that the trajectories of the system are confined within the physically meaningful domain $A$, we now identify the set of equilibrium points and characterize their stability.

%\begin{definition}
%    We define set $A^*$ as: $$A^*= \{(x, 0): x \in [0,1]\} \cup \Big\{\Big(1, \Big[1 - \frac{\gamma}{\beta}\Big]_+\Big)\Big\}.$$
%\end{definition}

\begin{proposition}\label{prop:equilibria}
\begin{enumerate}
\item[(i)] If $\beta \le \gamma$, the set of equilibrium points of \eqref{sys_dynamics} in $A$ is given by
\begin{equation}\label{eq:stable_below}
\{(x,0) \in A: x \in [0,1]\}\,,
\end{equation}
and all equilibria are stable but not asymptotically stable. 
\item[(ii)] If $\beta > \gamma$, the set of equilibrium points of \eqref{sys_dynamics} in $A$ is given by 
$$\{(x,0) \in A: x \in [0,1]\} \cup \{(1,y^*)\}\,, \quad y^* = 1-\frac{\gamma}{\beta}\,.$$
Moreover, the equilibrium points 
\begin{equation}\label{eq:stable}\Big\{(x,0) \in A: x \in \Big[0,\frac\gamma\beta\Big)\Big\}\end{equation}
are stable but not asymptotically stable, the equilibrium points
\begin{equation}\label{eq:unstable}\Big\{(x,0) \in A: x \in \Big[\frac\gamma\beta,1\Big]\Big\}\end{equation}
are unstable, and the equilibrium point
\begin{equation}\label{eq:as_stable}
(1,y^*)
\end{equation}
is asymptotically stable.
\end{enumerate}
\end{proposition}

\begin{proof}
    $(x, y)$ is an equilibrium point if and only if $\dot{x} = 0$ and $\dot{y} = 0$, that is, if
    \begin{equation*}
        \begin{cases}
            \dot{x} = (1 - x)y = 0 \\
            \dot{y} = \beta(x - y)y - \gamma y = 0\,.
        \end{cases}
    \end{equation*}  
    The first equation implies that either $y = 0$ or $x = 1$. If $y = 0$, then the second equation is satisfied for every $x$, which means that all points $(x,0)$ are equilibria. If $x = 1$, then the second equation yields $\beta(1 - y)y - \gamma y = 0$, which admits one solution $y = 0$ if $\beta \le \gamma$ or two distinct solutions $y = 0$ and $y = 1 - \gamma/\beta$ if $\beta > \gamma$. This completely characterizes the set of equilibrium points.

We now analyze the stability of the equilibrium points. First notice that the Jacobian of the system is
$$J(x,y) = \begin{pmatrix}
	-y & 1-x \\
	\beta y & \beta(x-2y) - \gamma
\end{pmatrix}\,.$$
Consider an equilibrium point $(x^*,0)$ with $x^* < \gamma/\beta$. Hence, the Jacobian matrix
\begin{equation}\label{eq:J0}
J(x^*,0) = \begin{pmatrix}
    0 & 1 - x^* \\
    0 & \beta x^* - \gamma
\end{pmatrix}\,,
\end{equation}
admits eigenvalues $\lambda_1 = 0$ and $\lambda_2 = \beta x^* - \gamma < 0$. We establish the stability of the equilibrium points by using the center manifold theorem \cite[Section 3.2]{Guckenheimer1983}. The center subspace for the equilibrium $(x^*,0)$ is $\mathcal V_c = \{(x,0): x \in \mathbb R\}$.
On this manifold, the dynamics is $\dot{x}=0$, implying the stability of the equilibrium point. This proves that the equilibrium points \eqref{eq:stable} when $\beta > \gamma$ and the equilibrium points \eqref{eq:stable_below} when $\beta \le \gamma$ are stable, excluding equilibrium $(1,0)$ when $\beta = \gamma$. However, $(1,0)$ is stable when $\beta = \gamma$, since in such a case $\dot x \ge 0$ and $\dot y \le 0$ for every $(x,y)$ in $A$. Furthermore, since \eqref{eq:stable_below} and \eqref{eq:stable} form a continuum of equilibrium, they cannot be asymptotically stable equilibrium points. This completely characterizes the stability of the equilibrium points \eqref{eq:stable_below}-\eqref{eq:stable}.
The instability of equilibrium points $(x^*,0)$ with $x^* > \gamma /\beta$ when $\beta > \gamma$ follows from the fact that the Jacobian matrix \eqref{eq:J0} has a positive eigenvalue $\beta x^* - \gamma$.
To conclude the proof of instability of the equilibrium points \eqref{eq:unstable},
we are left with proving the instability of $(\gamma/\beta,0)$. To this end, take $\delta>0$ and choose initial condition
$$
x(0)=\frac\gamma\beta+\frac{\delta}{2}, \quad y(0)=\frac{\delta}{2}.
$$
Then
$$
x(0)-y(0)=\frac{\gamma}{\beta},
$$
and therefore
$$
\dot y(0)
=
\beta(x(0)-y(0))y(0)-\gamma y(0)
=0.
$$
Moreover,
$$
\dot x(0)
=
(1-x(0))y(0)
=
\left(1-\frac{\gamma}{\beta}-\frac{\delta}{2}\right)\frac{\delta}{2}>0
$$
for sufficiently small $\delta>0$. Hence, for sufficiently small $t>0$,
$$
x(t)-y(t)>\frac{\gamma}{\beta},
$$
which implies
$$
\dot y(t) =\beta(x(t)-y(t))-\gamma y(t)>0\,.
$$
Consequently, $y(t)$ increases away from zero, and trajectories starting arbitrarily close to $(\gamma/\beta,0)$ leave any sufficiently small neighborhood of the equilibrium point. Hence, $(\gamma/\beta,0)$ is unstable, and the characterization of the stability of equilibrium points \eqref{eq:unstable} is thus complete.
        
        Finally, to prove the asymptotic stability of the equilibrium point $(1,y^*)$ when $\beta > \gamma$, note that
        $$
        J(1,y^*) = \begin{pmatrix}
        	\gamma/\beta - 1 & 0 \\
        	\beta - \gamma & \gamma - \beta
        \end{pmatrix}\,,
        $$ 
        hence both eigenvalues are real and negative, and the asymptotic stability follows from the linearization theorem.
\end{proof}

The set of equilibrium points and their stability is illustrated in Figure \ref{fig:comparison}. The next section provides a non-trivial invariant of motion of the dynamics and characterizes its level sets.
\section{Invariant of Motion}\label{sec:H}

To further characterize the global phase portrait and the long-term evolution of the social transition, we seek to identify the underlying constraints that govern the system trajectories. We shall see that the system admits an invariant of motion, providing a powerful tool for the system analysis.

\begin{definition}
Given $\beta>0$, let $H_{\beta}(x, y): [0,1) \times [0,1] \to \R$ defined by
$$H_{\beta}(x,y) = 
    \begin{cases}
        \displaystyle\frac{y + \gamma - 1}{1 - x} - \ln(1 - x) & \beta = 1 \\[7pt]
        \Big(y + \displaystyle\frac{\gamma}{\beta} - 1\Big)(1 - x)^{- \beta} - \frac{\beta}{1 - \beta}(1 - x)^{1 - \beta} & \beta \neq 1\,.
    \end{cases}$$ 
\end{definition}

\begin{remark}\label{remark:H}
    Note that $H_{\beta}(1,y)$ is not well defined for every $\beta > 0$. Moreover, $\partial H_\beta(x,y) / \partial y > 0$ for every $\beta > 0$, $x$ in $[0,1)$, and $y$ in $[0,1]$.
\end{remark}

\begin{proposition}\label{prop:H}
    $H_{\beta}(x,y)$ is an invariant of motion of model \eqref{sys_dynamics}, i.e., $\dot H_\beta(x,y) = 0$.
\end{proposition}

\begin{proof}
    We compute the time derivative of $H_{\beta}(x,y)$.
    If $\beta = 1$,
    \begin{equation*}            
        \begin{aligned}
            \dot{H}_{\beta}(x,y) &= \frac{\partial H_{\beta}}{\partial x}\dot{x} + \frac{\partial H_{\beta}}{\partial y}\dot{y} \\
            &= \bigg(\frac{y + \gamma - 1}{(1 - x)^2}  + \frac{1}{1 - x}\bigg)(1-x)y + \bigg(\frac{1}{1 - x} \bigg)\big((x-y)y-\gamma y\big) \\
            &= \bigg(\frac{y - x + \gamma}{(1 - x)^2} \bigg)(1-x)y + \bigg(\frac{1}{1 - x} \bigg)\big((x-y)y-\gamma y\big) \\
            &= \frac{y}{1 - x}\big((y - x + \gamma) + (x - y - \gamma)\big) \\[3pt]
            &= 0.
        \end{aligned}
    \end{equation*}
    If $\beta \neq 1$,
    \begin{equation*}            
        \begin{aligned}
            \dot{H}_{\beta}(x,y) &= \frac{\partial H_{\beta}}{\partial x}\dot{x} + \frac{\partial H_{\beta}}{\partial y}\dot{y} \\
            &= \bigg(\beta\Big(y+\displaystyle\frac{\gamma}{\beta}-1\Big)(1-x)^{-(\beta+1)}+\beta(1-x)^{-\beta}\bigg) (1-x)y \ + \\
            & + (1 - x)^{-\beta}\big(\beta(x-y)y-\gamma y\big) \\[3pt]
            &= (1-x)^{-\beta}y\big(\beta y+\gamma-\beta+\beta(1-x)+\beta(x-y)-\gamma\big)\\[3pt]
            &= 0.
        \end{aligned}
    \end{equation*}
    This concludes the proof.
\end{proof}

The next technical result characterizes the level sets of the invariant of motion above the threshold. Towards this goal,  for $\varepsilon \ge 0$, we let $g_\varepsilon: [0,1] \to \R$ 
%\textcolor{red}{be careful to the point $x=1$} 
defined by
    \begin{equation}\label{eq:g}
        g_{\varepsilon}(x) = \begin{cases}
            1 - \gamma + (1 - x)\Big(\ln{\displaystyle\frac{1-x}{1-\gamma}-1+\varepsilon}\Big) & \beta = 1\\[7pt]
            1 - \displaystyle\frac{\gamma}{\beta}+\frac{\beta}{1 - \beta}(1 - x)-C_{\varepsilon,\beta}(1-x)^{\beta} & \beta \neq 1\,,
        \end{cases}
    \end{equation}
with
\begin{equation}\label{eq:C}
C_{\varepsilon,\beta} = \frac{(1-\gamma/\beta)^{1-\beta}}{1-\beta}-\varepsilon
\end{equation}
and with convention $0 \cdot \log 0 = 0$.

%\textcolor{red}{I could actually modify the statement, because (i) is true and meaningful for every $\varepsilon$, in (ii) I might be interested in saying when $g_\varepsilon$ is convex or not (and it's quite easy). (iii) and (iv) are actually needed just for small $\varepsilon$}
\begin{lemma}\label{lemma:g}
Let $\beta > \gamma$. Then: 
\begin{enumerate}
    \item[(i)] given $\varepsilon \ge 0$, for every $(x,y)$ in $A$ with $x < 1$, $H_\beta(x,y) = H_\beta(\gamma/\beta,0) + \varepsilon$ if and only if $y = g_\varepsilon(x)$;

    \item[(ii)] if $\beta \ge 1$, $g_\varepsilon(x)$ is strictly convex for every $\varepsilon \ge 0$. If $\beta < 1$, let
    $$
    \varepsilon_\beta = \frac{(1-\gamma/\beta)^{1-\beta}}{1-\beta} = -H_\beta(\gamma/\beta,0)>0\,.
    $$
Then, $g_\varepsilon(x)$ is strictly convex if $\varepsilon< \varepsilon_\beta$, affine if $\varepsilon = \varepsilon_\beta$, and strictly concave if $\varepsilon > \varepsilon_\beta$;

    \item[(iii)]
    for every $\varepsilon \ge 0$ such that $g_\varepsilon(x)$ is convex, $g_\varepsilon(x)$ always admits a unique stationary point $(\hat x_\varepsilon,\hat y_\varepsilon)$ in $A$ given by
    \begin{equation}\label{eq:xy}    
        \hat{x}_{\varepsilon} =\begin{cases}
        1 - (1 - \gamma)e^{- \varepsilon}, & \beta = 1\\
        %1 - \displaystyle\frac{1}{\big((1 - \gamma/\beta)^{1 - \beta}-\varepsilon(1 - \beta)\big)^{1/(\beta - 1)}} 
        1-((1-\beta)C_{\varepsilon,\beta})^{1/(1-\beta)}
        & \beta \neq 1\,,
        \end{cases}\qquad \hat{y}_{\varepsilon} = \hat x_\varepsilon - \frac{\gamma}{\beta}\,.
\end{equation}
Moreover, for every $\varepsilon \ge 0$
\begin{equation}\label{eq:mono}
\frac{\partial \hat x_\varepsilon}{\partial \varepsilon} = \frac{\partial \hat y_\varepsilon}{\partial \varepsilon} > 0\,;
\end{equation}
        \item[(iv)] for every $\varepsilon \ge 0$ such that $g_{\epsilon}$ is convex,
        \begin{equation*}
            \hat y_\varepsilon \leq \Big(1 - \displaystyle\frac{\gamma}{\beta}\Big)^{\beta}\varepsilon\,;
        \end{equation*}
        \item[(v)] for a sufficiently small $\varepsilon \ge 0$, there exists finite $K_\varepsilon>0$ such that
        $$
        g_{\varepsilon}(x) \leq g_{\varepsilon}(\hat{x}_{\varepsilon}) + K_\varepsilon (x-\hat{x}_{\varepsilon})^2\,, \quad \forall x \in [\hat{x}_{\varepsilon} - \sqrt{\varepsilon}, \hat{x}_{\varepsilon} + \sqrt{\varepsilon}]\,.
        $$
    \end{enumerate}
\end{lemma}

\begin{proof}
See the Appendix.
\end{proof}

\section{Asymptotic and Transient Behavior}\label{sec:beh}

Building upon the existence of an invariant of motion and the geometry of its level sets, we now broaden our scope to investigate the global convergence properties of the system and the transient behavior of the dynamics.

\subsection{Asymptotic Behavior}\label{sec:asymptotic}
While local analysis characterizes the behavior near specific fixed points, it is imperative to determine the ultimate fate of the population state from any arbitrary initial condition within the feasible domain. We shall consider two cases: below the threshold ($\beta \le \gamma$) and above the threshold ($\beta > \gamma$). We start by considering the former case.

\begin{proposition}
Let $\beta \le \gamma$. Then, for every initial condition $(x(0),y(0)) = (x_0,y_0)$ in $A$: 
\begin{enumerate}
    \item[(i)] if $x_0 < 1$, then
\begin{equation}\label{eq:as_below}
(x(t),y(t)) \xrightarrow{t \to + \infty} (x^*,0)
\end{equation}
where $(x^*,0)$ is the unique point in $A$ such that $H_\beta(x_0,y_0) = H_\beta(x^*,0)$;
\item[(ii)] if $x_0 = 1$, then
\begin{equation}\label{eq:as_below2}
(x(t),y(t)) \xrightarrow{t \to + \infty} (1,0)\,.\end{equation}
\end{enumerate} 
\end{proposition}
\begin{proof}
It follows from Proposition \ref{proposition:basic}(ii) that $\dot x \ge 0$ and, since $\beta \le \gamma$, $\dot y \le 0$. Hence, $(x(t),y(t))$ must converge to an equilibrium point. Recall that all equilibrium points are in the form $(x^*,0)$ by Proposition \ref{prop:equilibria}(i). If $x_0 = 1$, then the claim follows from Proposition \ref{proposition:basic}(ii), since $x(t) = 1$ for every time $t \ge 0$, hence in particular $x^* = 1$. Consider now the case $x_0 < 1$, and note that $\beta \le \gamma$ implies $\partial H_\beta(x,y) / \partial x > 0$ for every $(x,y)$ in $A$ with $x < 1$. Hence, there exists a unique point $(x^*,0)$ such that $H_\beta(x_0,y_0) = H_\beta(x^*,0)$ and the dynamics must converge to such equilibrium point due to Proposition \ref{prop:H}.  
\end{proof}

\begin{figure}[thpb]
    \centering
    \begin{subfigure}{0.48\columnwidth}
        \centering
        \includegraphics[width=\linewidth]{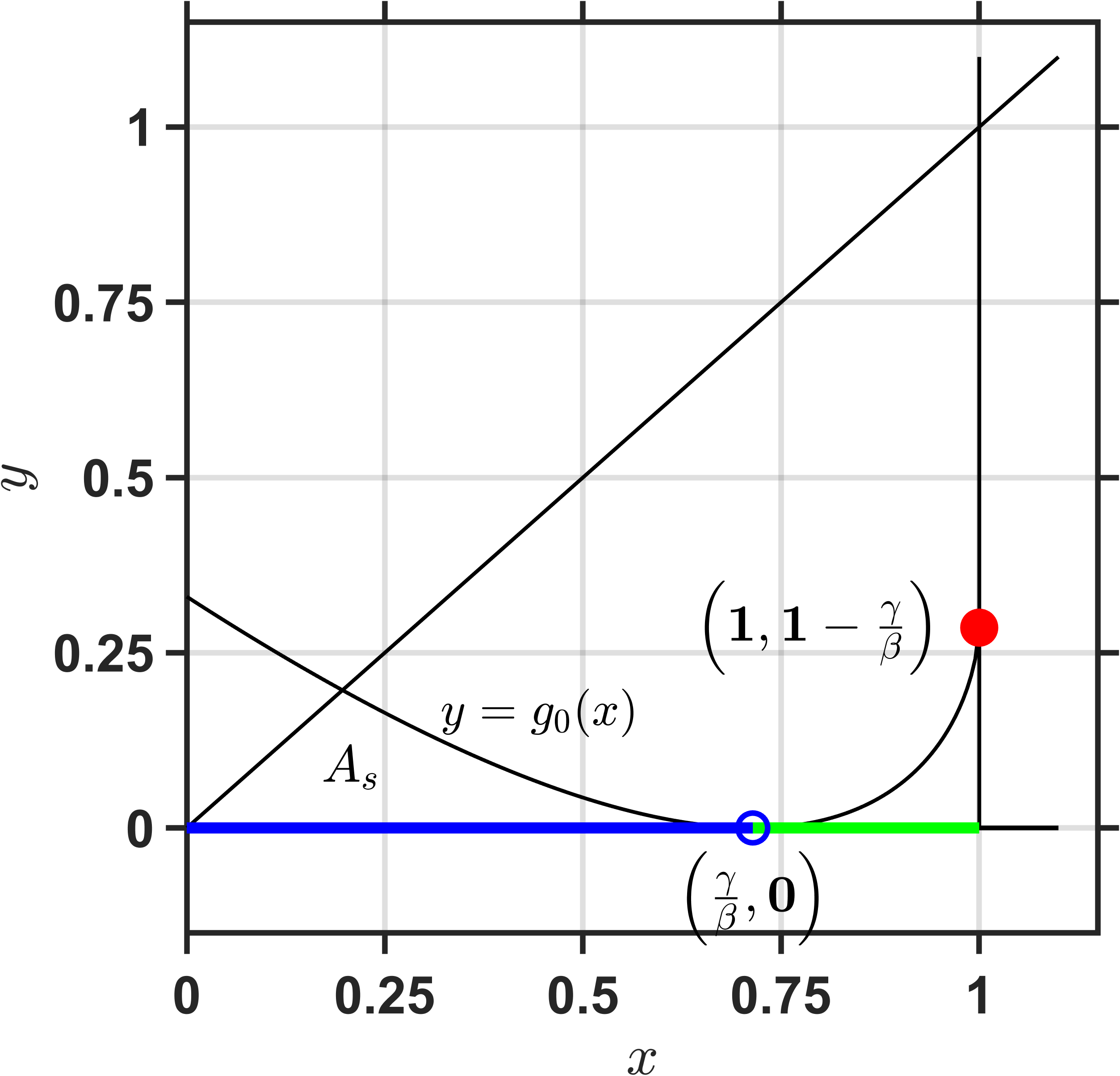}
        \caption{$\gamma < \beta$}
    \end{subfigure}
    \hfill
    \begin{subfigure}{0.48\columnwidth}
        \centering
        \includegraphics[width=\linewidth]{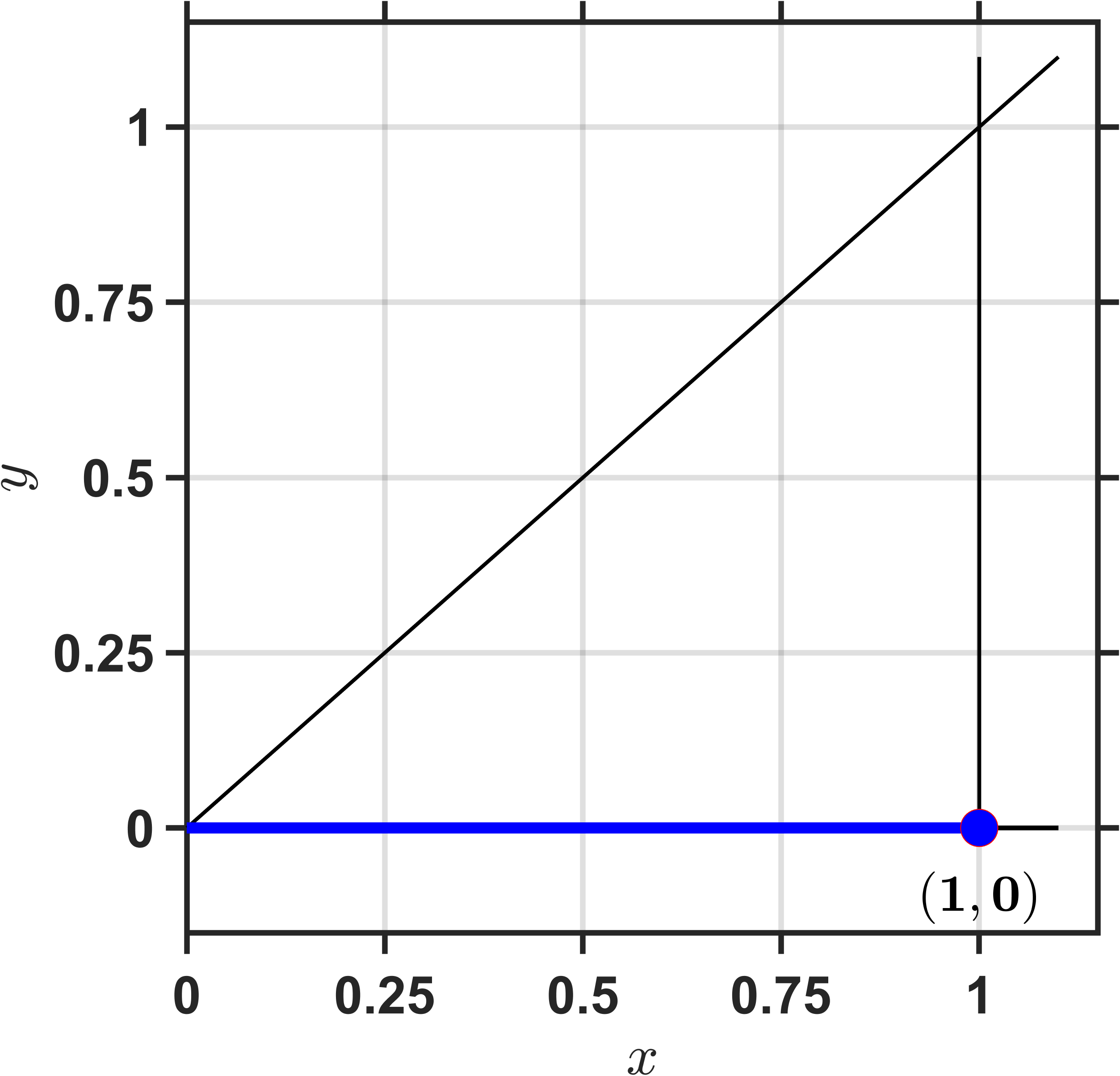}
        \caption{$\gamma \geq \beta$}
    \end{subfigure}
    \caption{Phase space diagrams below and above the critical threshold. For the parameter regime above the threshold, we illustrate equilibrium points with different stability properties, the curve $g_0$, and distinct dynamical regions. The blue equilibria are stable but not asymptotically stable, the red equilibria are asymptotically stable, and the green equilibria are unstable. \label{fig:comparison}}
\end{figure}

We now focus on the asymptotic system behavior above the threshold. 
To this end, we define the following subset of $A$.

\begin{definition}
Given $\beta > \gamma$, let
$$A_s = \Big\{(x, y)\in A: x \leq \frac{\gamma}{\beta}, \ y \le g_0(x)\Big\} \cup \{(x, y)\in A: y = 0\}\,.$$
\end{definition}

%\begin{figure}[thpb]
%    \centering
%    \includegraphics[width=0.35\textwidth]{set_extended_abstract.png}  
%    \caption{Geometric partitioning of the domain $A$ into regions $B$, $C$, and $D$. \textcolor{red}{no longer meaningful}}
%    \label{fig:basins_partition}  
%\end{figure}

\begin{theorem}
    Let $\beta > \gamma$. Then, for every initial condition $(x(0),y(0)) = (x_0,y_0)$ in $A$:
    \begin{enumerate}[label=(\roman*)]
        \item if $(x_0, y_0)\in A_s$, then $$(x(t), y(t)) \xrightarrow{t \to + \infty} (x^*, 0)\,,$$ where $x^*$ is the smallest solution of $H_\beta(x^*, 0) = H_\beta(x_0, y_0)$ if $y_0>0$, and $x^* = x_0$ if $y_0 = 0$.
        \item If $(x_0, y_0) \in A \setminus A_s$, then $$(x(t), y(t)) \xrightarrow{t \to + \infty} (1, y^*)\,, \qquad y^* = \displaystyle 1-\frac{\gamma}{\beta}\,.$$
    \end{enumerate}
\end{theorem}

\begin{proof}
Since $x(t)$ is monotonically non-decreasing by Proposition \ref{proposition:basic}(ii), then $x(t)$ admits a limit $x^*$ for every initial condition $(x_0,y_0)$. 

(i) If $y_0 = 0$, then $(x_0,y_0)$ is an equilibrium and the claim is trivial. Assume now that $y_0>0$. If $(x_0,y_0)$ belongs to $A_s$, we have that $y_0 < g_0(x_0)$, and $(x(t),y(t))$ can leave $A_s$ only by crossing the curve $y = g_0(x)$, which is prevented by Proposition \ref{prop:H} and Lemma \ref{lemma:g}(i) with $\varepsilon = 0$. If instead $(x_0,y_0)$ lies on the curve $y = g_0(x)$, then $(x(t),y(t))$ moves along the curve until convergence to the equilibrium point $(\gamma/\beta,0)$. This yields $x^* \le \gamma / \beta$. It then follows from Proposition \ref{proposition:basic}(ii) that $x(t) \le \gamma/\beta$ for every $t \ge 0$, hence, from \eqref{sys_dynamics}, $\dot y(t) \le 0$. Therefore, $(x(t),y(t))$ must converge to an equilibrium point $(x^*,0)$ in $A_s$. In particular, since $x(t)$ is continuous and non-decreasing in time, and since every element $(x,0)$ is an equilibrium point, $x^*$ is the smallest element greater than $x_0$ such that $H_\beta(x^*,0) = H_\beta(x_0,y_0)$ due to Proposition \ref{prop:H}.

(ii) Consider the sets $$C = \{(x,y) \in A: y > g_0(x)\}, \quad D = \{(x,y) \in A: 0 < y \le g_0(x), \ x > \gamma/\beta\}$$
and note that $A \setminus A_s = C \cup D$. We now prove that every initial condition in $C$ and $D$ converges to $(1,x^*)$. 

Consider $(x_0, y_0)$ in $C$. Since the trajectory $(x(t),y(t))$ cannot intercept the curve $y = g_0(x)$, then $C$ is positively invariant. We now argue by contradiction. Assume that $x^* \le \gamma /\beta$. By the same arguments used to prove item (i), this would imply that the limit point is in the form $(x^*,0)$, contradicting the fact that $C$ is positively invariant. Hence, $x^* > \gamma /\beta$.
Fix now a small $\delta>0$ such that
$$
x^*-\delta>\frac{\gamma}{\beta}.
$$
Since $x(t)\to x^*$, there exists $T_\delta>0$ such that
$$
x^* - \delta \le x(t) \le x^*\,,
\qquad \forall t\ge T_\delta\,.
$$
Hence, for every $t\ge T_\delta$,
$$
\big(\beta(x^*-\delta-y(t))-\gamma\big)y(t) \le \dot y(t)
\le
\big(\beta(x^*-y(t))-\gamma\big)y(t).
$$
It thus follows that
$$
x^*-\delta-\frac{\gamma}{\beta}
\le
\liminf_{t\to+\infty}y(t)
\le
\limsup_{t\to+\infty}y(t)
\le
x^*-\frac{\gamma}{\beta}.
$$
Letting $\delta\to0$, we conclude that
$$
\lim_{t \to + \infty} (x(t),y(t)) = \Big(x^*,x^*-\frac{\gamma}{\beta}\Big).$$
Since the only equilibrium point with positive second component is $(1,1-\gamma/\beta),
$
we conclude that $x^*=1$, and therefore
$$
\lim_{t \to +\infty} (x(t),y(t))
=
\left(1,1-\frac{\gamma}{\beta}\right) = (1,y^*)\,,
$$
concluding the proof of the case $(x_0,y_0)$ in $C$.

Finally, consider the case where $(x_0,y_0)$ belongs to $D$. Note that $D$ is positively invariant, and $\dot x,\dot y \ge 0$ for every $(x,y)$ in $D$. This proves that $x(t),y(t)$ admits a limit point and such limit point must be $(1,y^*)$ by Proposition \ref{prop:equilibria}.
\end{proof}

%\begin{figure}[H]
    %\centering
    
   % \includegraphics[width=0.7\textwidth]%{vector_field.png}  
  %  \caption{Vector field of the system}
 %   \label{fig:vector_field}
%\end{figure}

 %   \begin{figure}[H]
 %       \centering
        %\includegraphics[width=0.7\textwidth]{gamma_greater_beta_vary_C.png}
        %\caption{Level curves of the invariant of motion}
        %\label{fig:level_curve}
    %\end{figure}

\subsection{Transient Behavior}
While Section \ref{sec:asymptotic} completely determines the asymptotic behavior of the system through the invariant of motion $H_\beta$, the same invariant also provides detailed information on the transient evolution. In this section we exploit the geometry of the level sets of $H_\beta$ to characterize both the incubation phenomenon and the monotonicity of the advocacy level.

The incubation phenomenon refers to a prolonged period during which $y$ remains very close to zero, and the variation in $x$ is minimal. Following this phase, both 
$x$ and $y$ experience a sudden and rapid increase, ultimately converging to an equilibrium point. This phenomenon is characterized by a slow buildup followed by a sharp transition towards stability, which is shown in Figure \ref{fig:incubation_plot}. 

\begin{figure}[htb]
    \centering
    \includegraphics[width=8.4cm]{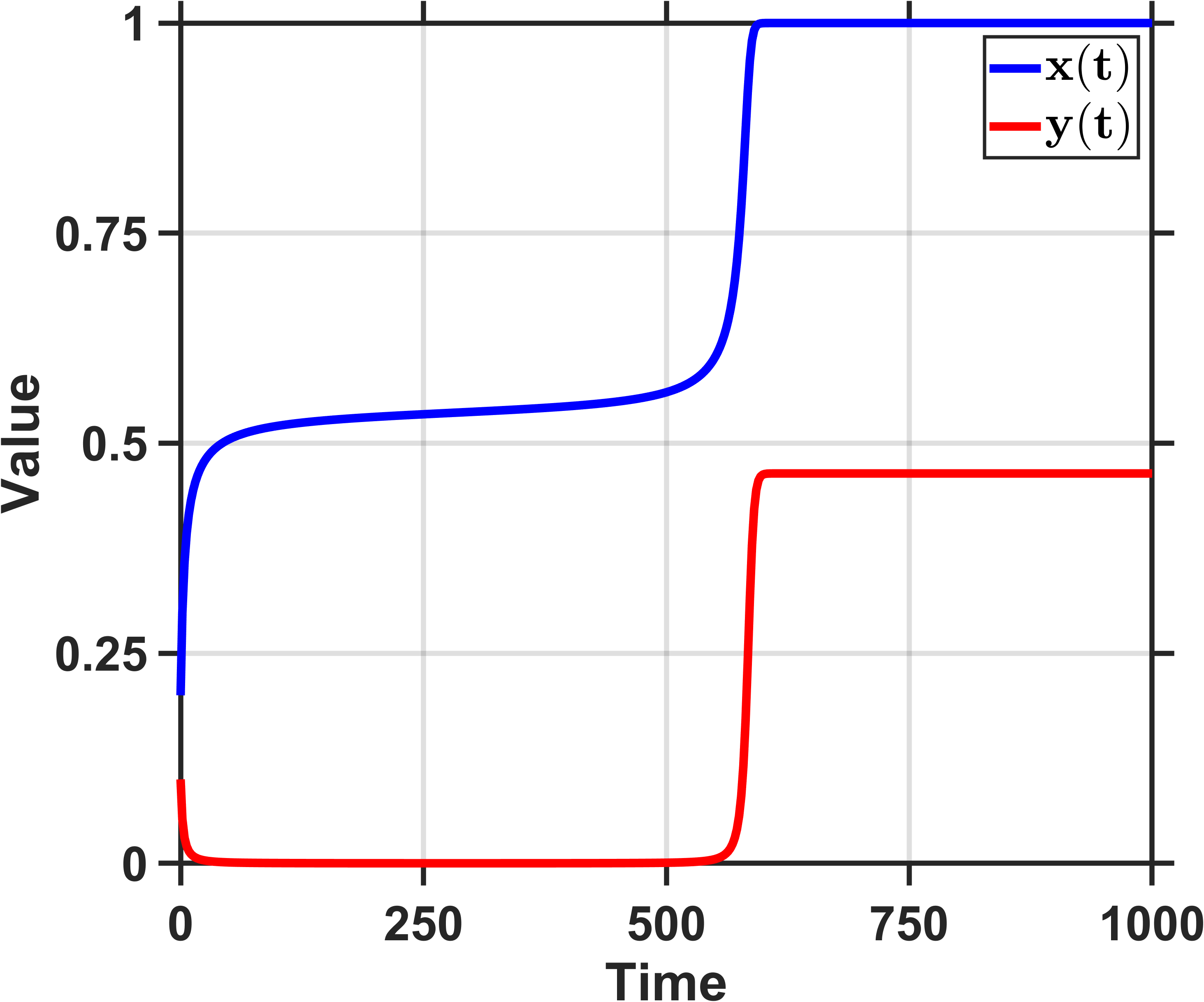} 
    \caption{A typical trajectory of $(x(t),y(t))$ exhibiting the incubation-outbreak pattern. The vocal advocacy $y(t)$ exhibits a U-shaped trajectory: initial decay is followed by a long quiescent period, ending in a sudden explosion as latent belief $x(t)$ crosses the critical value $\hat x_0 = \gamma/\beta$.}
    \label{fig:incubation_plot}
\end{figure}

It is apparent that the incubation phenomenon occurs only above the threshold, i.e., if $\beta > \gamma$. 
The next result provides a rigorous proof demonstrating that the duration of the incubation period diverges to infinity when the initial condition lies on the curve $y = g_\varepsilon(x)$ for small values of $\varepsilon > 0$. Towards this goal, given a small enough $\varepsilon > 0$, we  
define two time instants $t^-_{\varepsilon}$, $t^+_{\varepsilon}$ such that
$$x(t^-_{\varepsilon}) = \hat x_\varepsilon -\sqrt\varepsilon, \quad x(t^+_{\varepsilon}) = \hat x_\varepsilon + \sqrt\varepsilon\,,$$
whose uniqueness follows from the monotonicity of $x(t)$.

\begin{theorem}\label{thm:incubation}
Let $\beta > \gamma$. Then, given a small enough $\varepsilon > 0$ and an initial condition $(x_0,g_\varepsilon(x_0))$ with $x_0 \le \hat x_\varepsilon - \sqrt{\varepsilon}$, the relation
$$t^+_{\varepsilon} - t^-_{\varepsilon} \geq \frac{G_\varepsilon}{\sqrt{\varepsilon}}\,,$$ 
holds true, where $$G_\varepsilon = \displaystyle\frac{2}{(1 - \gamma/\beta)^{\beta} + K_\varepsilon}\,,$$
is bounded away from zero as $\varepsilon$ vanishes.
\end{theorem}

\begin{proof}
Note that
    \begin{align*}
        2\sqrt{\varepsilon} &= x(t^+_{\varepsilon}) - x(t^-_{\varepsilon}) = \int_{t^-_{\varepsilon}}^{t^+_{\varepsilon}}\dot{x}(t)dt = \int_{t^-_{\varepsilon}}^{t^+_\varepsilon}(1 - x(t))y(t)dt \leq \int_{t^-_{\varepsilon}}^{t^+_{\varepsilon}}y(t)dt \\
            & = \int_{t^-_{\varepsilon}}^{t^+_{\varepsilon}}g_{\varepsilon}(x(t))dt \leq \int_{t^-_{\varepsilon}}^{t^+_{\varepsilon}}(g_{\varepsilon}(\hat{x}_\varepsilon)+K_{\varepsilon}(x(t)-\hat{x}_{\varepsilon})^2)dt \\[5pt]
            & \leq (t^+_{\varepsilon}-t^-_{\varepsilon})\Big(1 - \displaystyle\frac{\gamma}{\beta}\Big)^{\beta}\varepsilon + K_{\varepsilon}(t^+_{\varepsilon}-t^-_{\varepsilon})\varepsilon.
    \end{align*}
    where the fourth equality follows from Lemma \ref{lemma:g}(i) and Proposition \ref{prop:H}, the second inequality follows from Lemma \ref{lemma:g}(v), and the last inequality from Lemma \ref{lemma:g}(iv) and from Proposition \ref{proposition:basic}(ii), which implies that $x(t) \in [\hat x_\varepsilon - \sqrt{\varepsilon}, \hat x_\varepsilon + \sqrt{\varepsilon}]$ for every time $t \in [t_\varepsilon^-,t_\varepsilon^+]$. Inverting this formula we get the statement.
\end{proof}

Theorem \ref{thm:incubation} provides a lower bound for the time needed to go from $\hat x_{\varepsilon} - \sqrt\varepsilon$ to $\hat x_{\varepsilon} + \sqrt\varepsilon$ supporters for initial conditions on the curve $y = g_\varepsilon(x)$. Note that such lower bound diverges as $\varepsilon$ tends to vanish due to the fact that $K_{\varepsilon}$ is bounded for every $\varepsilon \ge 0$ (cf. Lemma \ref{lemma:g}(v)), hence the trajectory spends at least $O(\varepsilon^{-1/2})$ time in a $O(\sqrt\varepsilon)$ neighborhood of the stationary point.

We now investigate changes in the monotonicity in the advocacy level. We shall consider the case above the threshold, since below the threshold the advocacy level is non-increasing for every $(x,y)$ in $A$.

%As $\dot{y} = \beta(x - y)y - \gamma y$, we have $\dot{y} > 0$ when $y < x - \displaystyle\frac{\gamma}{\beta}$, $\dot{y} = 0$ when $y = x - \displaystyle\frac{\gamma}{\beta}$ and $\dot{y} < 0$ when $y > x - \displaystyle\frac{\gamma}{\beta}$.

\begin{definition}
    We define the following subsets of $A$:
    $$
    \begin{aligned}
    A_+ & = \Big\{(x,y)\in A: 0 < y < x - \displaystyle\frac{\gamma}{\beta}\Big\}\,, \\
    A_- & = 
    \begin{cases}
    \{(x,y) \in A: y \ge g_{\varepsilon_\beta}(x)\} \quad & \beta < 1\,,\\
    \emptyset & \beta \ge 1\,,
    \end{cases}\\
\bar A & = A \setminus (A_s \cup A_+ \cup A_-)\,.
\end{aligned}
$$
\end{definition}

The three sets are illustrated in Figure \ref{fig:phase_partition}.We recall that the curve $g_{\varepsilon_\beta}(x)$ is affine, as proved in Lemma \ref{lemma:g}(ii). Moreover, the segment $y = x - \gamma/\beta$ identifies the stationary points of the curve $g_\varepsilon(x)$ as $\varepsilon$ varies, as Lemma \ref{lemma:g}(iii) proves. The next result characterizes the transient behavior of $y(t)$ for the initial conditions in each set.
%those three sets, whose union coincides with the basin of attraction of the equilibrium $(1,y^*)$.

\begin{figure}[thpb]
    \centering
    \begin{subfigure}{0.48\columnwidth}
        \centering
        \includegraphics[width=\linewidth]{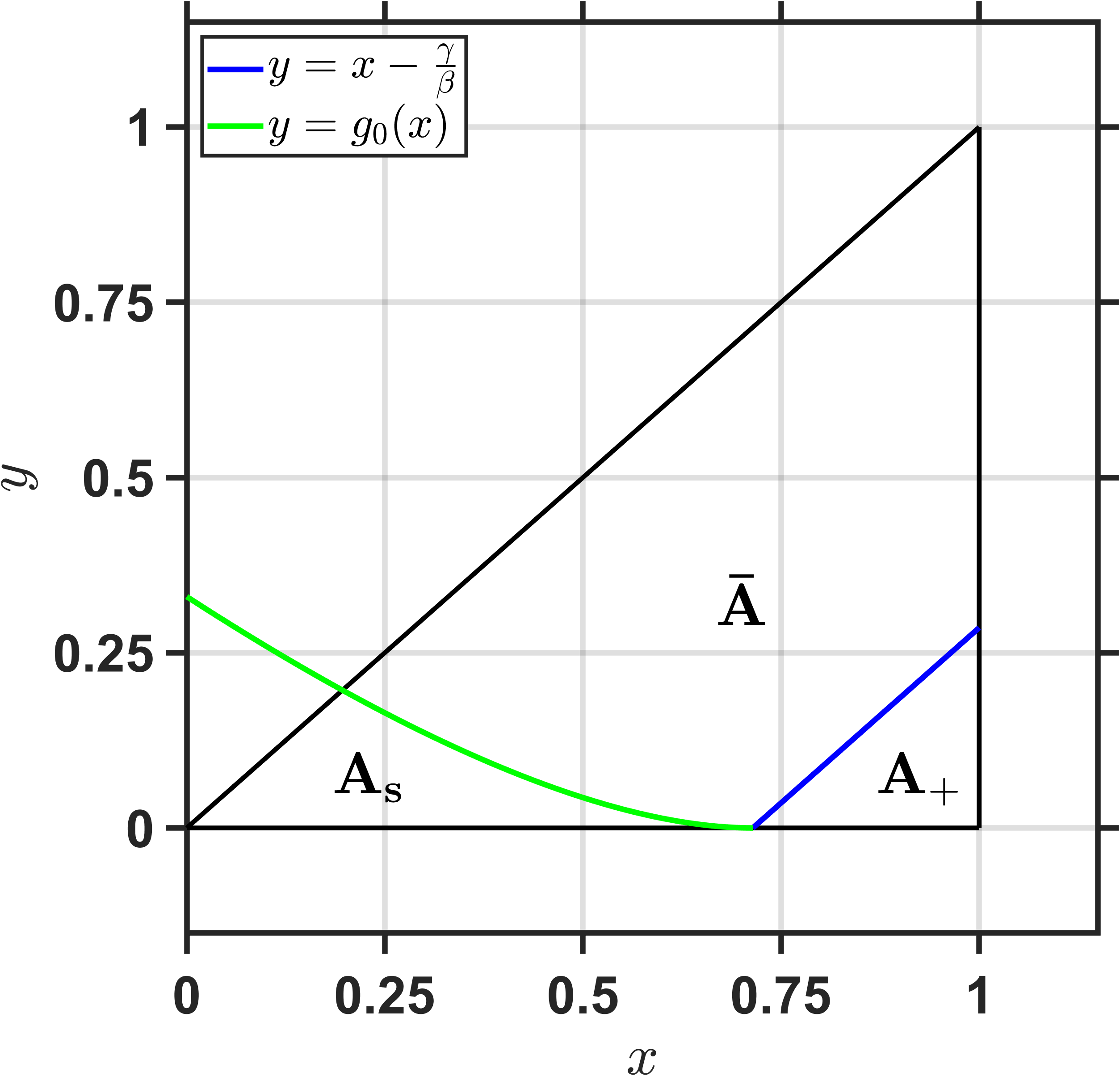}
        \caption{$\beta \ge 1$}
        \label{subfig:beta_ge_1}
    \end{subfigure}
    \hfill
    \begin{subfigure}{0.48\columnwidth}
        \centering
        \includegraphics[width=\linewidth]{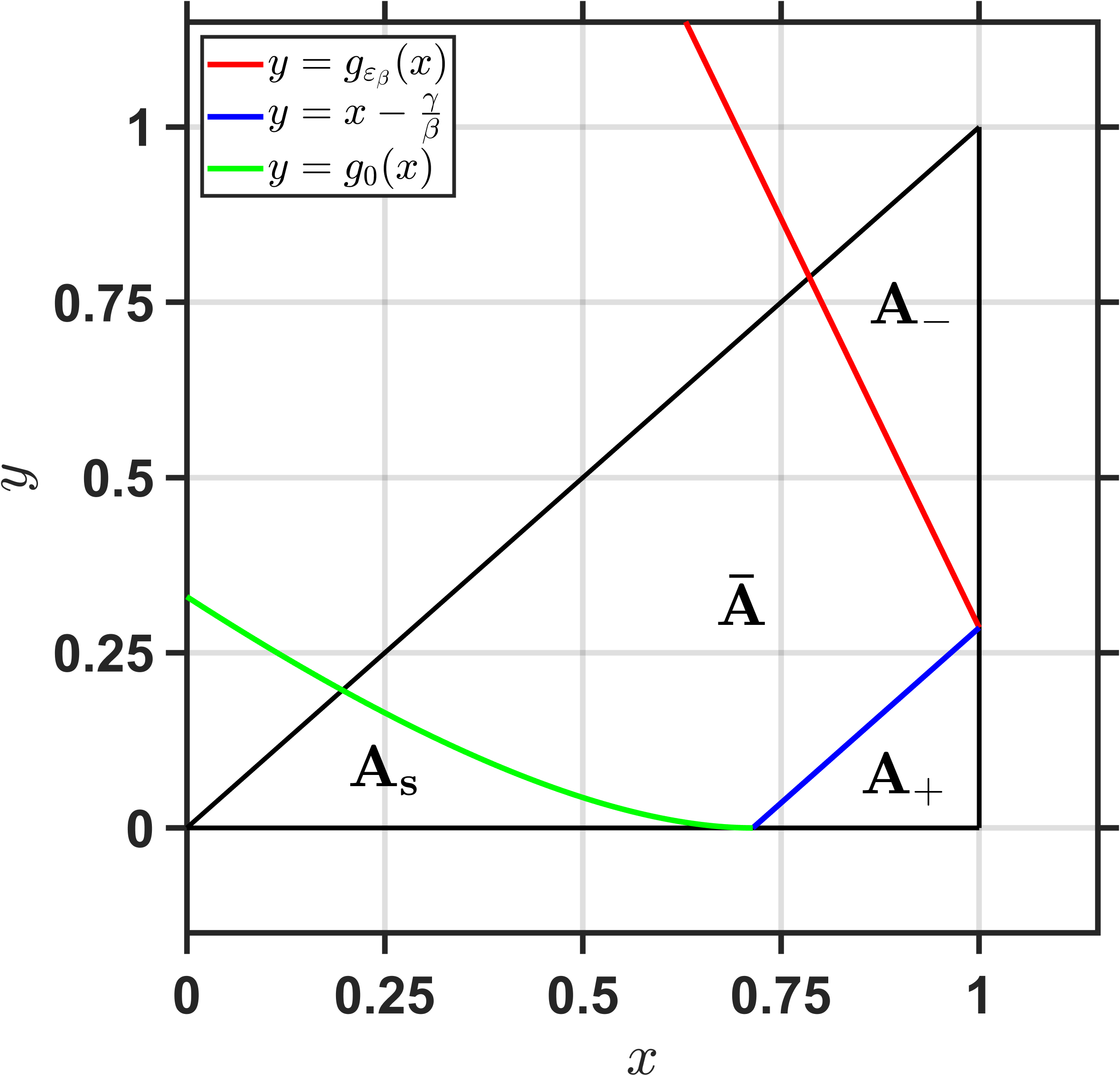}
        \caption{$\beta < 1$}
        \label{subfig:beta_lt_1}
    \end{subfigure}
    %\caption{Phase partitioning of the state space. Left ($\beta \ge 1$): The regions $I_1$ and $J$ are separated by the switching line $y = x - \gamma/\beta$. Right ($\beta < 1$): The region $K$ appears where trajectories converge to the endemic equilibrium monotonically. The green curve $y=g_0(x)$ denotes the boundary of the attraction basin of equilibrium $(1,y^*)$. \textcolor{red}{inconsistent notation, symbols are now different. I would also change the legent for the red line, which should be $y = g_{\varepsilon(x)}$}}
    \caption{Phase partitioning of the state space. Left ($\beta \ge 1$): The regions $A_s$ and $A_+$ are separated from $\bar A$ by the curve $y = g_0(x)$ (green) and the segment $y = x - \gamma/\beta$ (blue), respectively. Right ($\beta < 1$): An additional region $A_-$ appears, bounded by the curve $y = g_{\varepsilon_\beta}(x)$ (red line).} %Here, $y = g_0(x)$ represents the boundary of the attraction basin for the equilibrium $(1, y^*)$.}
    \label{fig:phase_partition}
\end{figure}
    
\begin{proposition}\label{prop:y_der}
Let $\beta > \gamma$, and consider model \eqref{sys_dynamics} with initial condition $(x_0,y_0)$ in $A$. Then:
\begin{enumerate}
\item[(i)] if $(x_0,y_0)$ belongs to $A_+$, $y(t)$ is monotonically increasing;
\item[(ii)] if $(x_0,y_0)$ belongs to $A_-\cup A_s$, $y(t)$ is monotonically decreasing;
\item[(iii)] if $(x_0,y_0)$ belongs to $\bar A$, $y(t)$ is monotonically decreasing until reaching a minimum point, and then increasing
\end{enumerate}
\end{proposition}

\begin{proof}
(i) Observe that $\dot y, \dot x > 0$ in the interior of $A_+$, and $\dot y = 0, \dot x > 0$ on the segment that separates $A_+$ from $\bar A$. Hence, $A_+$ is positively invariant, and therefore $y(t)$ is monotonically increasing for all initial conditions in $A_+$.

(ii) First note that that $\dot y < 0$ for all points in $A_-\cup A_s$. Moreover, both $A_-$ and $A_s$ are positively invariant, since the trajectory cannot cross the segment $y = g_{\varepsilon_\beta}(x)$, since such curve is a level set of the invariant of motion $H_\beta$. This readily proves the claim. 

%(iii) Note that $\dot y < 0$ for all points $(x,y)$ in $\bar A$. Moreover, every point $(x,y)$ in $\bar A$ lies on a curve $y = g_{\varepsilon}(x)$, with $\epsilon$ in $(0,\varepsilon_\beta)$, since $\partial \hat y_\varepsilon / \partial \varepsilon > 0$ by Lemma \ref{lemma:g}(iii). Lemma \ref{lemma:g}(ii) then implies that the curve $g_\varepsilon(x)$ is convex and admits a stationary point $(\hat x_\varepsilon,\hat y_\varepsilon)$ on the segment that separates $\bar A$ from $A_+$.  Since $A_+$ is positively invariant and $\dot y > 0$ for all points in $A_+$, the statement is proved. 
(iii) Note that $\dot y < 0$ for all points $(x,y)$ in $\bar A$. Moreover, since $\bar A$ lies between the curves $g_{0}$ and $g_{\varepsilon_\beta}$, every point $(x,y)$ in $\bar A$ lies on a unique convex level set $y = g_{\varepsilon}(x)$, with $\epsilon$ in $(0,\varepsilon_\beta)$, since $\partial \hat y_\varepsilon / \partial \varepsilon > 0$ by Lemma \ref{lemma:g}(iii). Lemma \ref{lemma:g}(ii) then implies that the curve $g_\varepsilon(x)$ is convex and admits a stationary point $(\hat x_\varepsilon,\hat y_\varepsilon)$ on the segment that separates $\bar A$ from $A_+$. Since $x(t)$ is monotonically increasing and each convex level curve admits a unique stationary point, the trajectory must cross the stationary point exactly once before entering $A_
+$. Since $A_+$ is positively invariant and $\dot y > 0$ for all points in $A_+$, the statement is proved. 
\end{proof}

Proposition \ref{prop:y_der} characterizes the monotonicity of the advocacy level $y$. Note that, if $(x_0,y_0)$ belongs to either $A_+$, $A_-$, or $\bar A$, the advocacy level will eventually converge to $y^*$. In contrast, if the initial condition lies in $A_s$, the advocacy level converges monotonically to $0$.

\section{Conclusion}\label{sec:conclusion}

In this paper, we have proposed and analyzed a minimal dynamical model to explore the diffusion of social norms through the interplay of internal belief conversion and public  mobilization. By partitioning the population into non-supporters, silent supporters, and vocal advocates, our model captures the critical tension between social recruitment and activist fatigue. 

Our analytical results provide a rigorous characterization of the system long-term behavior. We identified the conditions under which a new norm successfully saturates the population or fades into stagnation, demonstrating that these outcomes are governed by the intrinsic ratios of mobilization and fatigue to the conversion rate. Crucially, we discovered and quantified an "incubation phenomenon," where a deceptive period of quiescence precedes an abrupt, endogenous explosion of public support. This finding offers a mechanistic explanation for how social movements can "store energy" invisibly before reaching a tipping point, even in the absence of external shocks. 

The theoretical framework presented here contributes to the understanding of social tipping points by revealing a purely internal route to sudden social change.  Future research could extend this minimal model by incorporating network heterogeneities or exploring the impact of competitive dynamics between multiple social norms. By bridging the gap between individual-level fatigue and macroscopic social transitions, this work provides a foundation for more nuanced modeling of collective behavior in complex social systems.

To further validate the proposed "incubation phenomenon" and the underlying dynamical mechanisms of this model, future research should pursue empirical investigations across diverse disciplines, particularly at the intersection of social and natural sciences.

In sociology, large-scale social media analytics could distinguish between silent supporters and vocal advocates to verify the model's predicted thresholds in phenomena such as "spiral of silence" reversals or the "crossing the chasm" phase of innovation adoption. Parallelly, biological systems offer high-resolution settings for testing these state-transition mechanisms, ranging from the "quorum sensing" threshold in microbial colonies to the "angiogenic switch" in dormant oncology cells and the pre-ictal phase in neurobiology. By quantitatively tracking the interplay between recruitment (mobilization) and metabolic costs (fatigue) in these diverse environments, researchers can calibrate the model's parameters ($\alpha, \beta, \gamma$), transforming this minimal framework into a robust predictive tool for characterizing nonlinear phase transitions in both human and natural complex systems.

\bibliographystyle{plain}  
%plain, IEEEtran, alpha
\bibliography{references} 

@article{nowak1990private,
  title     = {From private attitude to public opinion: A dynamic theory of social impact.},
  author    = {Nowak, Andrzej and Szamrej, Jacek and Latan{\'e}, Bibb},
  journal   = {Psychological review},
  volume    = {97},
  number    = {3},
  pages     = {362},
  year      = {1990},
  publisher = {American Psychological Association}
}

@book{hale,
	address = {Huntington, N.Y.},
	author = {Hale, J. K.},
	edition = {2nd},
	publisher = {Robert E. Krieger Publishing Co., Inc.},
	title = {Ordinary Differential Equations},
	year = {1980}}

@book{Guckenheimer1983,
	address = {New York, NY},
	author = {J. Guckenheimer and P. Holmes},
	isbn = {978-0-387-90819-9},
	publisher = {Springer},
	series = {Applied Mathematical Sciences},
	title = {Nonlinear Oscillations, Dynamical Systems, and Bifurcations of Vector Fields},
	volume = {42},
	year = {1983}}

@article{brown2018economic,
  title     = {Economic correlates of footbinding: Implications for the importance of Chinese daughters’ labor},
  author    = {Brown, Melissa J and Satterthwaite-Phillips, Damian},
  journal   = {PloS one},
  volume    = {13},
  number    = {9},
  pages     = {e0201337},
  year      = {2018},
  publisher = {Public Library of Science San Francisco, CA USA}
}

@article{916272ae-6b7a-3e62-b7e2-5747187dae7b,
  issn      = {00251909, 15265501},
  url       = {http://www.jstor.org/stable/2628128},
  author    = {Frank M. Bass},
  journal   = {Management Science},
  number    = {5},
  pages     = {215--227},
  publisher = {INFORMS},
  title     = {A New Product Growth for Model Consumer Durables},
  urldate   = {2025-08-20},
  volume    = {15},
  year      = {1969}
}

@article{alos2016inertia,
  title     = {Inertia and decision making},
  author    = {Al{\'o}s-Ferrer, Carlos and H{\"u}gelsch{\"a}fer, Sabine and Li, Jiahui},
  journal   = {Frontiers in psychology},
  volume    = {7},
  pages     = {169},
  year      = {2016},
  publisher = {Frontiers Media SA}
}

@article{MORE2019102,
  title    = {A SI model for social media influencer maximization},
  journal  = {Applied Computing and Informatics},
  volume   = {15},
  number   = {2},
  pages    = {102-108},
  year     = {2019},
  issn     = {2210-8327},
  doi      = {https://doi.org/10.1016/j.aci.2017.11.001},
  url      = {https://www.sciencedirect.com/science/article/pii/S221083271730162X},
  author   = {Jyoti Sunil More and Chelpa Lingam}
}

@article{ye2021collective,
  title     = {Collective patterns of social diffusion are shaped by individual inertia and trend-seeking},
  author    = {Ye, Mengbin and Zino, Lorenzo and Mlakar, {\v{Z}}an and Bolderdijk, Jan Willem and Risselada, Hans and Fennis, Bob M and Cao, Ming},
  journal   = {Nature Communications},
  volume    = {12},
  number    = {1},
  pages     = {5698},
  year      = {2021},
  publisher = {Nature Publishing Group UK London}
}

@article{jain_public_nodate,
  title    = {Public Opinion and the Dynamics of Reform},
  author   = {Jain, Sanjay and Mukand, Sharun},
  langid   = {english}
}

@article{ravindran2014antecedents,
  title     = {Antecedents and effects of social network fatigue},
  author    = {Ravindran, Thara and Yeow Kuan, Alton Chua and Hoe Lian, Dion Goh},
  journal   = {Journal of the Association for Information Science and Technology},
  volume    = {65},
  number    = {11},
  pages     = {2306--2320},
  year      = {2014},
  publisher = {Wiley Online Library}
}

@article{kuran1991now,
  title     = {Now out of never: The element of surprise in the East European revolution of 1989},
  author    = {Kuran, Timur},
  journal   = {World politics},
  volume    = {44},
  number    = {1},
  pages     = {7--48},
  year      = {1991},
  publisher = {Cambridge University Press}
}

@article{PhysRevLett.86.2909,
  title     = {Small World Effect in an Epidemiological Model},
  author    = {Kuperman, Marcelo and Abramson, Guillermo},
  journal   = {Phys. Rev. Lett.},
  volume    = {86},
  issue     = {13},
  pages     = {2909--2912},
  numpages  = {0},
  year      = {2001},
  month     = {Mar},
  publisher = {American Physical Society},
  doi       = {10.1103/PhysRevLett.86.2909},
  url       = {https://link.aps.org/doi/10.1103/PhysRevLett.86.2909}
}

@article{XIAO2001733,
  title    = {Analysis of a Three Species Eco-Epidemiological Model},
  journal  = {Journal of Mathematical Analysis and Applications},
  volume   = {258},
  number   = {2},
  pages    = {733-754},
  year     = {2001},
  issn     = {0022-247X},
  doi      = {https://doi.org/10.1006/jmaa.2001.7514},
  url      = {https://www.sciencedirect.com/science/article/pii/S0022247X01975146},
  author   = {Yanni Xiao and Lansun Chen}
}

@article{PIQUEIRA2009355,
  title    = {A modified epidemiological model for computer viruses},
  journal  = {Applied Mathematics and Computation},
  volume   = {213},
  number   = {2},
  pages    = {355-360},
  year     = {2009},
  issn     = {0096-3003},
  doi      = {https://doi.org/10.1016/j.amc.2009.03.023},
  url      = {https://www.sciencedirect.com/science/article/pii/S0096300309002379},
  author   = {José Roberto C. Piqueira and Vanessa O. Araujo}
}

@article{annurev:/content/journals/10.1146/annurev.phyto.45.062806.094357,
  author    = {Gilligan, Christopher A. and van den Bosch, Frank},
  title     = {Epidemiological Models for Invasion and Persistence of Pathogens},
  journal   = {Annual Review of Phytopathology},
  year      = {2008},
  volume    = {46},
  number    = {Volume 46, 2008},
  pages     = {385-418},
  doi       = {https://doi.org/10.1146/annurev.phyto.45.062806.094357},
  url       = {https://www.annualreviews.org/content/journals/10.1146/annurev.phyto.45.062806.094357},
  publisher = {Annual Reviews},
  issn      = {1545-2107},
  type      = {Journal Article}
}

@article{https://doi.org/10.1155/2013/721406,
  author   = {Rodrigues, Helena Sofia and Monteiro, M. Teresa T. and Torres, Delfim F. M.},
  title    = {Sensitivity Analysis in a Dengue Epidemiological Model},
  journal  = {Conference Papers in Science},
  volume   = {2013},
  number   = {1},
  pages    = {721406},
  doi      = {https://doi.org/10.1155/2013/721406},
  url      = {https://onlinelibrary.wiley.com/doi/abs/10.1155/2013/721406},
  eprint   = {https://onlinelibrary.wiley.com/doi/pdf/10.1155/2013/721406},
  year     = {2013}
}

@article{BETTENCOURT2006513,
  title    = {The power of a good idea: Quantitative modeling of the spread of ideas from epidemiological models},
  journal  = {Physica A: Statistical Mechanics and its Applications},
  volume   = {364},
  pages    = {513-536},
  year     = {2006},
  issn     = {0378-4371},
  doi      = {https://doi.org/10.1016/j.physa.2005.08.083},
  url      = {https://www.sciencedirect.com/science/article/pii/S0378437105010113},
  author   = {Luís M.A. Bettencourt and Ariel Cintrón-Arias and David I. Kaiser and Carlos Castillo-Chávez}
}

@article{doi:10.1177/0885728810387922,
  author  = {Robert W. Flexer and Alfred W. DavisoIII and Robert M. Baer and Rachel McMahan Queen and Richard S. Meindl},
  title   = {An Epidemiological Model of Transition and Postschool Outcomes},
  journal = {Career Development for Exceptional Individuals},
  volume  = {34},
  number  = {2},
  pages   = {83-94},
  year    = {2011},
  doi     = {10.1177/0885728810387922},
  url     = { https://doi.org/10.1177/0885728810387922},
  eprint  = {https://doi.org/10.1177/0885728810387922}
}

@article{doi:10.1068/b2833,
  author   = {Ling Bian},
  title    = {A Conceptual Framework for an Individual-Based Spatially Explicit Epidemiological Model},
  journal  = {Environment and Planning B: Planning and Design},
  volume   = {31},
  number   = {3},
  pages    = {381-395},
  year     = {2004},
  doi      = {10.1068/b2833},
  url      = { https://doi.org/10.1068/b2833},
  eprint   = { https://doi.org/10.1068/b2833}
}

@article{doi:10.1073/pnas.2311584120,
  author   = {Chadi M. Saad-Roy  and Arne Traulsen },
  title    = {Dynamics in a behavioral–epidemiological model for individual adherence to a nonpharmaceutical intervention},
  journal  = {Proceedings of the National Academy of Sciences},
  volume   = {120},
  number   = {44},
  pages    = {e2311584120},
  year     = {2023},
  doi      = {10.1073/pnas.2311584120},
  url      = {https://www.pnas.org/doi/abs/10.1073/pnas.2311584120},
  eprint   = {https://www.pnas.org/doi/pdf/10.1073/pnas.2311584120}
}

@article{doi:10.1086/226707,
  author   = {Granovetter, Mark},
  title    = {Threshold Models of Collective Behavior},
  journal  = {American Journal of Sociology},
  volume   = {83},
  number   = {6},
  pages    = {1420-1443},
  year     = {1978},
  doi      = {10.1086/226707},
  url      = { https://doi.org/10.1086/226707},
  eprint   = { https://doi.org/10.1086/226707}
}

@article{doi:10.1126/science.aas8827,
  author   = {Damon Centola  and Joshua Becker  and Devon Brackbill  and Andrea Baronchelli },
  title    = {Experimental evidence for tipping points in social convention},
  journal  = {Science},
  volume   = {360},
  number   = {6393},
  pages    = {1116-1119},
  year     = {2018},
  doi      = {10.1126/science.aas8827},
  url      = {https://www.science.org/doi/abs/10.1126/science.aas8827},
  eprint   = {https://www.science.org/doi/pdf/10.1126/science.aas8827}
}

@article{doi:10.1073/pnas.1418838112,
  author   = {Damon Centola  and Andrea Baronchelli },
  title    = {The spontaneous emergence of conventions: An experimental study of cultural evolution},
  journal  = {Proceedings of the National Academy of Sciences},
  volume   = {112},
  number   = {7},
  pages    = {1989-1994},
  year     = {2015},
  doi      = {10.1073/pnas.1418838112},
  url      = {https://www.pnas.org/doi/abs/10.1073/pnas.1418838112},
  eprint   = {https://www.pnas.org/doi/pdf/10.1073/pnas.1418838112}
}

@article{10.1098/rsta.2021.0169,
  author   = {Roy, Subhadeep and Biswas, Soumyajyoti},
  title    = {Opinion dynamics: public and private},
  journal  = {Philosophical Transactions of the Royal Society A: Mathematical, Physical and Engineering Sciences},
  volume   = {380},
  number   = {2224},
  pages    = {20210169},
  year     = {2022},
  month    = {04},
  issn     = {1364-503X},
  doi      = {10.1098/rsta.2021.0169},
  url      = {https://doi.org/10.1098/rsta.2021.0169},
  eprint   = {https://royalsocietypublishing.org/rsta/article-pdf/doi/10.1098/rsta.2021.0169/1323858/rsta.2021.0169.pdf}
}

@article{8353159,
  author   = {Pagliara, Renato and Dey, Biswadip and Leonard, Naomi Ehrich},
  journal  = {IEEE Control Systems Letters},
  title    = {Bistability and Resurgent Epidemics in Reinfection Models},
  year     = {2018},
  volume   = {2},
  number   = {2},
  pages    = {290-295},
  doi      = {10.1109/LCSYS.2018.2832063}
}

\appendix

\section{Proof of Lemma \ref{lemma:g}}
Let $\beta = 1$ and consider a point $(x,y)$ in $A$ with $x<1$ and $y = g_\varepsilon(x)$. Then,
$$\begin{aligned}
        H_{\beta}(x,y) &= H_{\beta}(x,g_{\varepsilon}(x)) \\[5pt]    
        & = \displaystyle\frac{g_{\varepsilon}(x) + \gamma - 1}{1 - x} - \ln(1 - x)\\
        &= \displaystyle\frac{1 - \gamma + (1 - x)\big(\ln (1-x) - \ln (1-\gamma) -1 + \varepsilon\big) + \gamma - 1}{1 - x} - \ln(1 - x)\\
        &= -1-\ln(1 - \gamma) + \varepsilon \\
        &= H_\beta(\gamma/\beta,0) + \varepsilon\,.
    \end{aligned}$$
Let now $\beta \neq 1$ and consider a point $(x,y)$ in $A$ with $x<1$ and $y = g_\varepsilon(x)$. Then,
    $$\begin{aligned}
        H_{\beta}(x,y) &= H_{\beta}(x,g_{\varepsilon}(x)) \\
        &= \Big(g_{\varepsilon}(x) + \displaystyle\frac{\gamma}{\beta} - 1\Big)(1 - x)^{-\beta} - \frac{\beta}{1 - \beta}(1 - x)^{1 - \beta}\\
        &= \Bigg(1 - \displaystyle\frac{\gamma}{\beta} + \frac{\beta}{1 - \beta}(1 - x) - \bigg(\frac{(1-\frac{\gamma}{\beta})^{1-\beta}}{1-\beta}-\varepsilon\bigg)(1-x)^{\beta}+\frac{\gamma}{\beta}-1\Bigg)(1 - x)^{-\beta} \ + \\
        &-\frac{\beta}{1 - \beta}(1 - x)^{1 - \beta}\\[5pt]
        &= - \displaystyle\frac{(1-\frac{\gamma}{\beta})^{1-\beta}}{1-\beta} + \varepsilon \\[5pt]
        &= H_\beta(\gamma/\beta,0) + \varepsilon\,.
    \end{aligned}$$
This proves the 'if' implication of item (i). The 'only if' implication follows from the fact that $\partial H_\beta(x,y) / \partial y > 0$ (cf. Remark \ref{remark:H}). 

To prove the other items, we compute
explicitly the first and second derivative of $g_\varepsilon(x)$ by
        \begin{equation}\label{eq:g'}   
            g'_{\varepsilon}(x) = \begin{cases}
                -\varepsilon -\ln{\displaystyle\frac{1-x}{1-\gamma}}, & \beta = 1\\[10pt]
                -\displaystyle\frac{\beta}{1 - \beta}+C_{\varepsilon,\beta}\beta(1-x)^{\beta-1} & \beta \neq 1\,,
            \end{cases}
        \end{equation}
        \begin{equation}\label{eq:g''}
            g''_{\varepsilon}(x) = \begin{cases}
                \displaystyle\frac{1}{1-x} & \beta = 1\\[10pt]
                C_{\varepsilon,\beta} \beta(1-\beta)(1 - x)^{\beta - 2} & \beta \neq 1\,.
            \end{cases}
        \end{equation}
        From this expression, it follows that, if $\beta = 1$, $g_\varepsilon$ is strictly convex for every $\varepsilon \ge 0$. if $\beta > 1$, we have from \eqref{eq:C} that $C_{\varepsilon,\beta}<0$, hence $g_\varepsilon$ is convex for every $\varepsilon \ge 0$. If $\beta<1$, we have that $C_{\varepsilon,\beta}>0$ if $\varepsilon < \varepsilon_\beta$, $C_{\varepsilon,\beta} = 0$ if $\varepsilon = \varepsilon_\beta$, and $C_{\varepsilon,\beta}<0$ if $\varepsilon > \varepsilon_\beta$. This yields item (ii).
        
        The expression of $\hat x_\varepsilon$ follows from imposing $g_\varepsilon'(\hat x_\varepsilon) = 0$. Moreover, when $\beta \neq 1$, it follows from $g_\varepsilon'(\hat x_\varepsilon) = 0$
        that
        $$
        (1-\hat x_\varepsilon)^{\beta-1} = \frac{1}{(1-\beta)C_{\varepsilon,\beta}}\,.
        $$
        Using this expression and the definition of $g_\varepsilon$ in \eqref{eq:g}, we have that
        $$
        \begin{aligned}
        \hat y_\varepsilon & = g_\varepsilon(\hat x_\varepsilon) = 1 - \frac{\gamma}{\beta} + \frac{\beta}{1-\beta}(1-\hat x_\varepsilon) - C_{\varepsilon,\beta}(1-\hat x_\varepsilon)^{\beta} \\
        & = 1 - \frac{\gamma}{\beta} + \frac{\beta}{1-\beta}(1-\hat x_\varepsilon) - \frac{1-\hat x_{\varepsilon}}{1-\beta} \\
        & = \hat x_\varepsilon - \frac{\gamma}{\beta}\,.
        \end{aligned}
        $$
        When $\beta = 1$,
        $$
        \hat y_\varepsilon = g_\varepsilon(\hat x_\varepsilon) = 1 - \gamma + (1-\gamma)e^{-\varepsilon}(-\varepsilon + 1 + \varepsilon) = \hat x_\varepsilon - \gamma\,,
        $$
        thus proving the expression of $\hat y_\varepsilon$. To prove its monotonicity in $\varepsilon$, note that
        \begin{equation}\label{der1:xy}
        \hat y_\varepsilon' = \hat x_\varepsilon' = \begin{cases}
        (1-\gamma) e^{-\varepsilon} \quad & \beta = 1\\
        ((1-\beta)C_{\varepsilon,\beta})^{\beta/(1-\beta)} & \beta \neq 1\,,
        \end{cases}
        \end{equation}
        and both expressions are positive when $\beta > \gamma$ and when $g_\varepsilon$ is strictly convex: in fact, when $\beta = 1$, the fact that $\beta > \gamma$ implies $1-\gamma>0$; when $\beta \neq 1$, $g_\varepsilon$ is strictly convex if and only if $(1-\beta)C_{\varepsilon,\beta}>0$, yielding the monotonicity of $\hat x_\epsilon$. It still remains to prove that $(\hat x_\varepsilon, \hat y_\varepsilon)$ belongs to $A$ for every $\varepsilon \ge 0$. To this end, note that 
        \begin{equation}\label{eq:limit_x}
        \hat x_0 = \gamma/\beta\,, \quad \lim_{\varepsilon \to +\infty} \hat x_\varepsilon = 1\,,
        \end{equation}
        so that $\hat x_\varepsilon \in [\gamma/\beta,1)$ for every $\varepsilon \ge 0$ because of its monotonicity. Finally, the rightmost equation in \eqref{eq:xy} implies that $\hat y_\varepsilon \le \hat x_{\varepsilon}$ for every $\varepsilon \ge 0$, and, together with \eqref{eq:limit_x}, that
        $$
        \hat y_0 = 0\,, \qquad \lim_{\varepsilon \to +\infty} \hat y_{\varepsilon} = 1-\frac{\gamma}{\beta}\,,
        $$
        thus showing that $(\hat x_\varepsilon,\hat y_\varepsilon)$ belongs to $A$ for every $\varepsilon \ge 0$ and concluding the proof of item (iii).

        To prove item (iv), we derive \eqref{der1:xy} and obtain 
        $$\hat y_\varepsilon'' = \hat x_\varepsilon'' = 
        \begin{cases}
        -(1-\gamma) e^{-\varepsilon} \quad & \beta = 1 \\
        -\beta ((1-\beta)C_{\varepsilon,\beta})^{(2\beta-1)/(1-\beta)} & \beta \neq 1\,,
        \end{cases}
        $$
showing that $\hat y_\varepsilon$, whenever the stationary point exists, is concave. Hence, evaluating \eqref{der1:xy} with $\varepsilon = 0$,
$$
\hat y_\varepsilon \le \hat y_0' \varepsilon = \Big(1-\frac \gamma \beta\Big)^{\beta} \varepsilon\,.
$$

We now prove item (v). First note that, for sufficiently small $\varepsilon \ge 0$, $g_{\varepsilon}$ is convex and admits a stationary point $(\hat x_{\varepsilon},\hat y_{\varepsilon})$ due to item (ii). According to the Taylor expansion with Lagrange Remainder we have that, given $\varepsilon \ge 0$, for every $x$ in $[\hat{x}_{\varepsilon} - \sqrt{\varepsilon}, \hat{x}_{\varepsilon} + \sqrt{\varepsilon}]$, there exists $\xi$ in $[\hat{x}_{\varepsilon} - \sqrt{\varepsilon}, \hat{x}_{\varepsilon} + \sqrt{\varepsilon}]$ such that
            \begin{align*}
                g_{\varepsilon}(x) &= g_{\varepsilon}(\hat{x}_{\varepsilon}) + g'_{\varepsilon}(\hat{x}_{\varepsilon})(x - \hat{x}_{\varepsilon}) + \frac{g''_{\varepsilon}(\xi)}{2}(x - \hat{x}_{\varepsilon})^2\\
                &\leq g_{\varepsilon}(\hat{x}_{\varepsilon}) + K_\varepsilon(x - \hat{x}_{\varepsilon})^2,
            \end{align*}
            where
            $$
            K_\varepsilon = \max_{x \in [\hat{x}_{\varepsilon} - \sqrt{\varepsilon}, \hat{x}_{\varepsilon} + \sqrt{\varepsilon}]} g_\varepsilon''(x) / 2
            $$
            and $g_\varepsilon'(\hat x_\varepsilon) = 0$.
            Now, note from \eqref{eq:g''} that $g_\varepsilon''(x)$ is bounded if $x$ is bounded away from $1$. Moreover, for sufficiently small $\varepsilon$, $x = 1$ is not included in the interval $[\hat x_{\varepsilon} - \sqrt{\varepsilon}, \hat x_{\varepsilon} + \sqrt{\varepsilon}]$. This proves that $K_{\varepsilon}$ is finite and concludes the proof.
            %If $\beta = 1$, it follows from \eqref{eq:g''} and from the expression of $\hat x_{\varepsilon}$ that
            %$$K_\varepsilon = \frac{1}{2(1 - \hat{x}_{\varepsilon} - \sqrt{\varepsilon})}, \qquad \lim_{\varepsilon \to 0^+} K_\varepsilon = \frac{1}{2(1-\gamma/\beta)} < +\infty\,,$$
            %and $K_\varepsilon>0$ for small $\varepsilon$.
            %If $\beta \neq 1$ and for small $\varepsilon \ge 0$, explicit computation yields that $g_\varepsilon''(x)$ is increasing if $\beta < 2$, constant if $\beta = 2$, and decreasing if $\beta > 2$.
            %From this observation and from \eqref{eq:g''} it follows that $$K_\varepsilon = 
            %\begin{cases}
            %    \displaystyle \frac{C_{\varepsilon,\beta}}{2}\beta(1-\beta)(1 - \hat{x}_{\varepsilon} - \sqrt{\varepsilon})^{\beta -2} & \text{if} \ \beta \leq 2\,, \\[10pt]
             %   \displaystyle \frac{C_{\varepsilon,\beta}}{2}\beta(1-\beta)(1 - \hat{x}_{\varepsilon} + \sqrt{\varepsilon})^{\beta -2} & \text{if} \ \beta > 2\,.
            %\end{cases}$$
%Hence, we observe that $K_\varepsilon>0$ for small enough $\varepsilon$, and
%$$
%\lim_{\varepsilon \to 0^+} K_\varepsilon = \frac{\beta}{2(1-\gamma/\beta)} < +\infty\,.
%$$
%This concludes the proof.
\end{document}